\documentclass[12pt, onecolumn]{IEEEtran}
\usepackage{subfigure, epsfig, color}
\usepackage{amsmath}
\usepackage{amsthm}
\usepackage{amssymb}
\usepackage{multirow}
\usepackage{booktabs}
\newtheorem{theorem}{Theorem}
\newtheorem{rem}{Remark}
\newtheorem{lemma}{Lemma}

\usepackage{mathtools}
\newcommand\myeqa{\stackrel{\mathclap{\footnotesize\mbox{(a)}}}{=}}
\newcommand\myeqb{\stackrel{\mathclap{\footnotesize\mbox{(b)}}}{=}}
\newcommand\myeqc{\stackrel{\mathclap{\footnotesize\mbox{(c)}}}{=}}
\makeatletter
\newcommand*{\rom}[1]{\expandafter\@slowromancap\romannumeral #1@}
\makeatother

\begin{document}

\title{Autoencoder-Based Error Correction Coding for One-Bit Quantization}
\author{{
    Eren Balevi and
    Jeffrey G. Andrews}\\
\thanks{The authors are with Department of Electrical and Computer Engineering at the University of Texas at Austin, TX, USA.  Email: erenbalevi@utexas.edu, jandrews@ece.utexas.edu.}
}

\maketitle 
\normalsize
\begin{abstract}
This paper proposes a novel deep learning-based error correction coding scheme for AWGN channels under the constraint of one-bit quantization in the receivers. Specifically, it is first shown that the optimum error correction code that minimizes the probability of bit error can be obtained by perfectly training a special autoencoder, in which ``perfectly'' refers to converging the global minima. However, perfect training is not possible in most cases. To approach the performance of a perfectly trained autoencoder with a suboptimum training, we propose utilizing turbo codes as an implicit regularization, i.e., using a concatenation of a turbo code and an autoencoder.  It is empirically shown that this design gives nearly the same performance as to the hypothetically perfectly trained autoencoder, and we also provide a theoretical proof of why that is so. The proposed coding method is as bandwidth efficient as the integrated (outer) turbo code, since the autoencoder exploits the excess bandwidth from pulse shaping and packs signals more intelligently thanks to sparsity in neural networks. Our results show that the proposed coding scheme at finite block lengths outperforms conventional turbo codes even for QPSK modulation. Furthermore, the proposed coding method can make one-bit quantization operational even for $16$-QAM.
\end{abstract}

\begin{IEEEkeywords}
Deep learning, error correction coding, one-bit quantization.
\end{IEEEkeywords}

\section{Introduction}
Wireless communication systems are trending towards ever higher carrier frequencies, due to the large bandwidths available \cite{TedInvitedPaper19}.  These high frequencies are made operational by the use of large co-phased antenna arrays to enable directional beamforming.  Digital control of these arrays is highly desirable, but requires a very large number of analog-to-digital converters (ADCs) at the receiver or digital-to-analog converters (DACs) at the transmitter, each of which consumes nontrivial power and implementation area \cite{Walden99}.  Low resolution quantization is thus inevitable to enable digital beamforming in future systems.  However, little is known about optimum communication techniques in a low resolution environment.  

In this paper, we focus on error correction codes for the one-bit quantized channel, where just the sign of the real and imaginary parts is recorded by the receiver ADC.  Conventional coding techniques, which mainly target unquantized additive white Gaussian noise (AWGN) channels or other idealized models, are not well-suited for this problem.  Deep learning is an interesting paradigm for developing channel codes for low-resolution quantization, motivated by its previous success for some other difficult problems, e.g., see \cite{OShea17} for learning transmit constellations, \cite{YeLi18} for joint channel estimation and data detection, \cite{BalAnd19} for one-bit OFDM communication, or \cite{MaoHao18} for several other problems.   This paper develops a novel approach which concatenates a conventional turbo code with a deep neural network -- specifically, an autoencoder -- to approach theoretical benchmarks and achieve compelling error probability performance.


\subsection{Related Work and Motivations}
Employing a neural network for decoding linear block codes was proposed in the late eighties \cite{BruckBlaum89}. Similarly, the Viterbi decoder was implemented with a neural network for convolutional codes in the late nineties \cite{WangWicker96}, \cite{HamalainenHenriksson99}. A simple classifier is learned in these studies instead of a decoding algorithm. This leads to a training dataset that must include all codewords, which makes them infeasible for most codes due to the exponential complexity. Recently, it was shown that a decoding algorithm could be learned for structured codes \cite{Gruber17}, however this design still requires a dataset with at least $90\%$ percent of the codebook, which limits its practicality to small block lengths. To learn decoding for large block lengths, \cite{KimViswanath18} trained a recurrent neural network for small block lengths that can generalize well for large block lengths. Furthermore, \cite{Nachmani18} improves the belief propagation algorithm by assigning trainable weights to the Tanner graph for high-density parity check (HDPC) codes that can be learned from a single codeword, which prevents the curse of dimensionality. 

Most of the prior studies are aimed at learning and/or improving the performance of decoding algorithms through the use of a neural network. There are only a few papers that aim to learn an encoder, which is more difficult than learning a decoder due to the difficulties of training the lower layers in deep networks \cite{KimViswanath18b}, \cite{JiangViswanath18}, \cite{Kosaian18}. We specifically design a channel code for the challenging one-bit quantized AWGN channels via an autoencoder to obtain reliable communication at the Shannon rate. The closest study to our paper that we know of is \cite{OShea17}, which proposes an autoencoder to learn transmit constellations such as M-QAM. However, \cite{OShea17} does not aim to achieve a very small error probability (close to the Shannon bound) and quantization is not considered.

\subsection{Contributions}
Our contributions are (i) to show that near-optimum hand-crafted channel codes can be equivalently obtained by perfectly training a special autoencoder, which however is not possible in practice, and (ii) to design a novel and practical autoencoder-based channel coding scheme that is well-suited for receivers with one-bit quantization.

 \textbf{Designing an optimum channel code is equivalent to learning an autoencoder}.  We first show that the mathematical model of a communication system can be represented by a regularized autoencoder, where the regularization comes from the channel and RF modules. Then, it is formally proven that an optimum channel code can be obtained by perfectly training the parameters of the encoder and decoder -- where ``perfectly'' means finding the global minimum of its loss function -- of a specially designed autoencoder architecture. However, autoencoders cannot be perfectly trained, so suboptimum training policies are utilized. This is particularly true for one-bit quantization, which further impedes training due to its zero gradient. Hence, we propose a suboptimum training method and justify its efficiency by theoretically finding the minimum required SNR level that yields almost zero detection error, which could be obtained if the autoencoder parameters would be trained perfectly, and prove the existence of a global minimum. This is needed, because we cannot empirically obtain the performance of a perfectly trained autoencoder due to getting stuck in a local minima. In what follows, observing the SNRs due to suboptimum training and comparing it with the case of perfect training allows us to characterize the efficiency. 

\textbf{Designing a practical coding scheme for one-bit receivers}.  Although one-bit quantization has been extensively studied, e.g., \cite{IvrlacNossek07}, \cite{JacobssonStuder17},  \cite{StuderDurisi16}, \cite{RisiLarsson14}, there is no paper to our knowledge that designs a channel code specifically for one-bit quantization. We fill this gap by developing a novel deep learning-based coding scheme that combines turbo codes with an autoencoder. Specifically, we first suboptimally train an autoencoder, and then integrate a turbo code with this autoencoder, which acts as an implicit regularizer. The proposed coding method is as bandwidth efficient as just using the turbo code, because the autoencoder packs the symbols intelligently by exploiting its sparsity stemming from the use of a rectified linear unit (ReLU) activation function and exploits the pulse shaping filter's excess bandwidth by using the faster-than-Nyquist transmission. It is worth emphasizing that conventional channel codes are designed according to the traditional orthogonal pulses with symbol rate sampling and cannot take the advantage of excess bandwidth. The numerical results show that our method can approach the performance of a perfectly trained autoencoder.  For example, the proposed coding scheme can compensate for the performance loss of QPSK modulation at finite block lengths due to the one-bit ADCs, and significantly improve the error rate in case of $16$ QAM, in which case one-bit quantization does not usually work even with powerful turbo codes. This success is theoretically explained by showing that the autoencoder produces Gaussian distributed data for turbo decoder even if there are some nonlinearities in the transmitters/receivers that result in non-Gaussian noise. 

This paper is organized as follows. The mathematical model of a communication system is introduced as a channel autoencoder in Section \ref{Channel Autoencoders}. Then, the training imperfections are quantified by finding the minimum required SNR level that achieves almost zero detection error for the one-bit quantized channel autoencoder in Section \ref{Quantifying}. The channel code is designed in Section \ref{Code Design}, and its performance is given in Section \ref{Numerical Results}. The paper concludes in Section \ref{Conclusions}.

\section{Channel Autoencoders} \label{Channel Autoencoders}
Autoencoders are a special type of feedforward neural network involving an ``encoder" that transforms the input message to a codeword via hidden layers and a ``decoder" that approximately reconstructs the input message at the output using the codeword.  This does not mean that autoencoders strive to copy the input message to the output. On the contrary, the aim of an autoencoder is to extract lower dimensional features of the inputs by hindering the trivial copying of inputs to outputs. Different types of regularization methods have been proposed for this purpose based on denoising \cite{VincentManzagol08}, sparsity \cite{RanzatoLeCun06}, and contraction \cite{RifaiBengio11}, which are termed \textit{regularized autoencoders}. A special type of regularized autoencoder inherently emerges in communication systems, where the physical channel as well as the RF modules of transmitters and receivers behave like a explicit regularizer. We refer to this structure as a \textit{channel autoencoder}, where channel refers to the type of regularization. 

The mathematical model of a communication system is a natural partner to the structure of a regularized autoencoder, since a communication system has the following ingredients:
\begin{enumerate}
\item A message set $\{1,2,\cdots,M\}$, in which message $i$ is drawn from this set with probability $1/M$
\item An encoder $f: \{1,2,\cdots,M \}\rightarrow X^n$ that yields length-$n$ codewords 
\item A channel $p(y|x)$ that takes an input from alphabet $X$ and outputs a symbol from alphabet $Y$
\item A decoder $g: Y^n\rightarrow \{1,2,\cdots,M \}$ that estimates the original message from the received length-$n$ sequence
\end{enumerate}
In regularized autoencoders, these 4 steps are performed as determining an input message, encoding this message, regularization, and decoding, respectively. To visualize this analogy, the conventional representation of a communication model is portrayed as an autoencoder that performs a classification task in Fig. \ref{fig:model}.
\begin{figure*} [!t] 
\centering 
\includegraphics [width=6.5in]{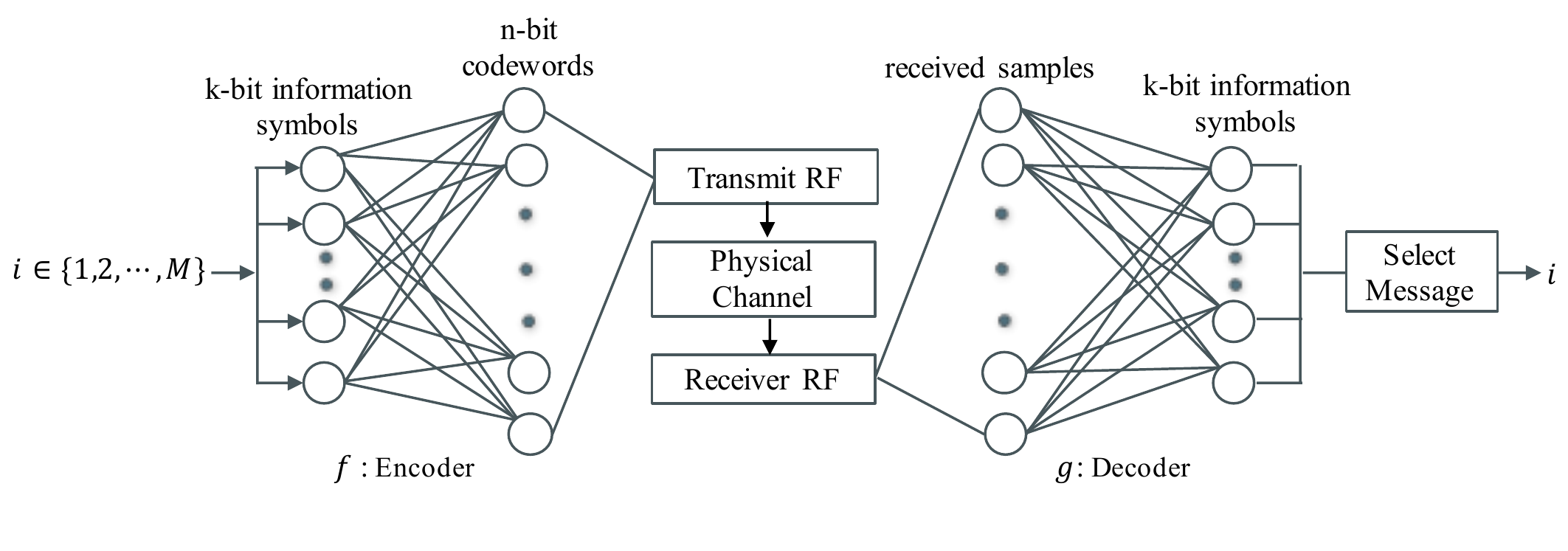}
\caption{Representation of a channel autoencoder: $i^{th}$ message is coded to $k$-bits information sequence, which is then mapped to a length-$n$ codeword via a parameterized encoder and transferred over the channel. The received signal is processed via a parameterized decoder to extract the message $i$.}\label{fig:model}
\end{figure*}

The fundamental distinction between a general regularized autoencoder and a communication system is that the former aims to learn useful features to make better classification/regression by sending messages, whereas the latter aims to minimize communication errors by designing hand-crafted features (codewords). This analogy is leveraged in this paper to design efficient coding methods by treating a communication system as a channel autoencoder for a challenging communication environment, in which designing a hand-crafted code is quite difficult. In this manner, we show that finding the optimum encoder-decoder pair with coding theory in the sense of minimum probability of bit error can give the same encoder-decoder pair that is learned through a regularized autoencoder.

An autoencoder aims to jointly learn a parameterized encoder-decoder pair by minimizing the reconstruction error at the output. That is,
\begin{equation}\label{deep_learning}
(f_{AE},g_{AE}) = \underset{f,g}{\mathrm{arg\ min\ }} J_{AE}(\theta_f,\theta_g)
\end{equation}
where $\theta_f$ and $\theta_g$ are the encoder and decoder parameters of $f:\mathbb{R}^k \rightarrow \mathbb{R}^n$ and $g:\mathbb{R}^n \rightarrow \mathbb{R}^k$ respectively, and
\begin{equation}\label{obj_fnc_ae}
J_{AE}(\theta_f,\theta_g)=\frac{1}{C}\sum_{c=1}^CL(\textbf{s}_{c},g(f(\textbf{s}_{c})))
\end{equation}
where $\textbf{s}_{c}$ is the input training vector and $C$ is the number of training samples. To find the best parameters that minimize the loss function, $L(\textbf{s}_{c},g(f(\textbf{s}_{c}))$ is defined as the negative log likelihood of $\textbf{s}_{c}$. The parameters are then trained through back-propagation and gradient descent using this loss function. The same optimization appears in a slightly different form in conventional communication theory. In this case, encoders and decoders are determined so as to minimize the transmission error probability given by
\begin{equation}\label{comm_theory}
(f^*,g^*) =\underset{f,g}{\mathrm{arg\ min\ }} \epsilon(n,M)
\end{equation}
where
\begin{equation}\label{error}
\epsilon(n,M) = \mathcal{P}[g(Y^n)\neq i | f(i)]
\end{equation} 
for a given $n$, $M$ and signal-to-noise-ratio (SNR). Note that \eqref{comm_theory} can be solved either by human ingenuity or a brute-force search. For the latter, if all possible combinations of mapping $2^k$ number of $k$-information bits to the $2^n$ codewords are observed by employing a maximum likelihood detection, the optimum linear block code can be found in terms of minimum probability of error. However, it is obvious that this is NP-hard. Thus, we propose an alternative autoencoder based method to solve \eqref{comm_theory}. 
\begin{theorem}\label{Theorem1}
The optimization problems in (\ref{deep_learning}) and (\ref{comm_theory}) are equivalent, i.e., they yield the same encoder-decoder pair for an autoencoder that has one-hot coding at the input layer and softmax activation function at the output layer, whose parameters are optimized by the cross entropy function.
\end{theorem}
\begin{proof}
See Appendix \ref{appendix-A}.
\end{proof}

\begin{rem}\label{Remark1}
Theorem \ref{Theorem1} states that a special autoencoder that is framed for the mathematical model of a communication system, which was defined in Shannon's coding theorem, can be used to obtain the optimum channel codes for any block length. This is quite important, because there is not any known tool that gives the optimum code as a result of the mathematical modeling of a communication system. Shannon's coding theorem only states that there is at least one good code without specifying what it is, and only for infinite block lengths. Hence, autoencoders can in principle be used for any kind of environment to find optimum error correction codes. However, the autoencoder must be perfectly trained, which is challenging or impossible.
\end{rem}

\section{Quantifying Training Imperfections in Channel Autoencoders} \label{Quantifying}
The channel autoencoder specified in Theorem \ref{Theorem1} would negate the need to design sophisticated hand-crafted channel codes for challenging communication environments, if it was trained perfectly. However, training an autoencoder is a difficult task, because of the high probability of getting stuck in a local minima. This can stem from many factors such as random initialization of parameters, selection of inappropriate activation functions, and the use of heuristics to adapt the learning rate. Handling these issues is in particular difficult for deep neural networks, which leads to highly suboptimum training and generalization error. Put differently, these are key reasons why deep neural networks were not successfully trained until the seminal work of \cite{HintonTeh06}, which proposed a greedy layerwise unsupervised pretraining for initialization. In addition to this, there were other improvements related to better understanding of activation functions, e.g., using a sigmoid activation function hinders the training of lower layers due to saturated units at the top hidden layers \cite{GlorotBengio10}. Despite these advances, there is still not any universal training policy that can guarantee to approach the global minimum, and using a suboptimum training, which usually converges to a local minima in optimizing the loss function, is inevitable.

To quantify how well a suboptimum training approach can perform, we need to know the performance of the perfectly trained autoencoder. However, finding this empirically is not possible due to getting stuck in one of the local minimas. Hence, we first find the minimum required SNR to have bit error probability approaching zero (in practice, less than $10^{-5}$).  Such a low classification error can usually be achieved only if the parameters satisfy the global minima of the loss function, corresponding to perfect training. Then, we quantify the training imperfections in terms of SNR loss with respect to this minimum SNR, which serves us as a benchmark. Since our main goal is to design channel codes for one-bit quantized AWGN channels, which is treated as a one-bit quantized AWGN channel autoencoder, this method is used to quantify the training performance of this autoencoder. Here, one-bit quantization enables us to save hardware complexity and power consumption for communication systems that utilize an ever-increasing number of antennas and bandwidth particularly at high carrier frequencies \cite{BalAnd19}. In the rest of this section, we first determine the minimum required SNR level for the one-bit quantized AWGN channel autoencoder in which the autoencoder can achieve zero classification error (or bit error rate) above this SNR, and then formally show there exists a global minimum and at least one set of encoder-decoder pair parameters converges to this global minimum.

\subsection{Minimum SNR for Reliable Coding for One-Bit Quantized Channel Autoencoders}
The encoder and decoder of the one-bit quantized AWGN channel autoencoder are parameterized via two separate hidden layers with a sufficient number of neurons (or width)\footnote{We prefer to use single layer with large number of neurons instead of multiple hidden layers with fewer neurons to make the analysis simpler and clearer without any loss of generality.}. To have a tractable analysis, a linear activation function is used at the encoder -- whereas there can be any nonlinear activation function in the decoder -- and there is a softmax activation function at the output. Since an autoencoder is trained with a global reconstruction error function, nonlinearities in the system can be captured thanks to the decoder even if the encoder portion is linear. 

To satisfy Theorem \ref{Theorem1}, one-hot coding is employed for the one-bit quantized AWGN channel autoencoder, which yields a multi-class classification. Specifically, the $i^{th}$ message from the message set $\{1,2,\cdots,M\}$ is first coded to the $k$-bit information sequence $\textbf{s}$. Then, $\textbf{s}$ is converted into $\textbf{x}$ using one-hot coding, and encoded with $f$, which yields an $n$-bit codeword. Adding the noise to this encoded signal produces the unquantized received signal, which is given by
\begin{equation}\label{unquantized_sig}
\textbf{y} = \theta_f\textbf{x} +\textbf{z}
\end{equation}
where $\textbf{z}$ is the additive Gaussian noise with zero mean and variance $\sigma_z^2$, and $\theta_f$ is the encoder parameters. Here, complex signals are expressed as a real signal by concatenating the real and imaginary part. Notice that there is a linear activation function in the encoder.

One-bit quantization, which is applied element-wise, constitutes the quantized received signal
\begin{equation}\label{quantized_sig}
\textbf{r} = \mathcal{Q}(\textbf{y}) = \text{sign}(\textbf{y}).
\end{equation}
The one-bit quantized received signal is processed by the decoder $g(\cdot)$ via the parameters $\theta_g$ followed by the softmax activation function, which leads to $\hat{x}_l = [g(\textbf{r})]_l$,
where the output vector is $\hat{\textbf{x}}=[\hat{x}_1 \cdots \hat{x}_{d}]$ such that $d=2^k$. The parameters of $\theta_f$ and $\theta_g$ are trained by minimizing the cross entropy function between the input and output layer. This can be equivalently considered as minimizing the distance between the empirical and predicted conditional distributions. Following that it is trivial to obtain the estimate of $\hat{\textbf{s}}$ from $\hat{\textbf{x}}$.

The mutual information between the input and output vector is equal to the channel capacity\footnote{Note that the encoder and decoder of the autoencoder is considered as a part of the wireless channel, because there is some randomness in the encoder and decoder stemming from the random initialization of parameters, which affects the capacity.} 
\begin{equation}\label{cap_exp}
C = \underset{ p(\textbf{s}) }{\mathrm{max\ }} I(\textbf{s};\hat{\textbf{s}}).
\end{equation}
Assuming that symbols are independent and identically distributed, $I(\textbf{s};\hat{\textbf{s}})$ can be simplified to
\begin{equation} \label{scalarMutInfo}
\begin{split}
I(\textbf{s};\hat{\textbf{s}})  & \myeqa  \sum_{i=1}^{k}H(s_i|s_{i-1}, \cdots, s_1) - \sum_{i=1}^{k}H(s_i|s_{i-1}, \cdots, s_1, \hat{s}_1, \cdots, \hat{s}_{k}) \\
&\myeqb \sum_{i=1}^{k}H(s_i) - \sum_{i=1}^{k}H(s_i|\hat{s}_i) \\    
&\myeqc kI(s_i;\hat{s}_i)
\end{split}
\end{equation}
where (a) is due to chain rule, (b) is due to independence and (c) comes from the identical distribution assumption. The capacity of the one-bit quantized AWGN channel autoencoders can then be readily found as
\begin{equation}
C = \underset{ {k\rightarrow \infty}}{\text{lim}} \underset{ p(\textbf{s}) }{\text{sup}} \frac{1}{k}I(\textbf{s},\hat{\textbf{s}}) = \underset{ p(\textbf{s}) }{\mathrm{max\ }} I(s_i,\hat{s}_i).
\end{equation}

It is not analytically tractable to express $I(s_i;\hat{s}_i)$ in closed-form due to the decoder that yields non-Gaussian noise. However, \eqref{cap_exp} can be equivalently expressed by replacing $I(\textbf{s};\hat{\textbf{s}})$ with $I(\textbf{s};\textbf{r})$ thanks to the data processing inequality, which qualitatively states that clever manipulations of data cannot enhance the inference, i.e., $I(\textbf{s};\hat{\textbf{s}})\leq I(\textbf{s};\textbf{r})$.
\begin{lemma}\label{Lemma2}
The mutual information between $\textbf{s}$ and $\textbf{r}$ in the case of a one-bit quantized channel autoencoder is 
\begin{eqnarray}
I(\textbf{s};\textbf{r}) & \leq & n\mathbb{E}_{\theta_f}[1+Q\left(\theta_f\sqrt{\gamma}\right)\log\left(Q\left(\theta_f\sqrt{\gamma}\right)\right)+\left(1-Q\left(\theta_f\sqrt{\gamma}\right)\right)\log\left(1-Q\left(\theta_f\sqrt{\gamma}\right)\right)]
\end{eqnarray}
where $\gamma$ is the transmit SNR and $Q(t)=\int_t^{\infty}\frac{1}{\sqrt{2\pi}}e^{-t^2/2}dt$
provided the encoder parameters are initialized with Gaussian random variables
\end{lemma}
\begin{proof}
See Appendix \ref{appendix-B}.
\end{proof}

It is worth emphasizing that the most common weight initialization in deep neural networks is to use Gaussian random variables \cite{GlorotBengio10}, \cite{HeSun15}. The minimum required SNR $\gamma_{\rm min}$ for the one-bit quantized AWGN channel autoencoder can be trivially found through Lemma \ref{Lemma2} when the code rate $R=\log_2(M)/n$ is equal to the capacity. That is, $\gamma_{\rm min} = \underset{ \{R=C \} }{\min}{\gamma}$. 
Specifically, the capacity is numerically evaluated in Fig. \ref{fig:capacity} using Lemma \ref{Lemma2} so as to determine the minimum required SNR to suppress the regularization impact for the one-bit quantized AWGN channel autoencoder. To illustrate, for a code rate of $\frac{1}{3}$, we find $\gamma_{\rm min} = 1.051$ dB. This means that if the one-bit quantized AWGN channel autoencoder is perfectly trained, it gives almost zero classification error above an SNR of $1.051$ dB.
\begin{figure} [!t] 
\centering 
\includegraphics [width=5in]{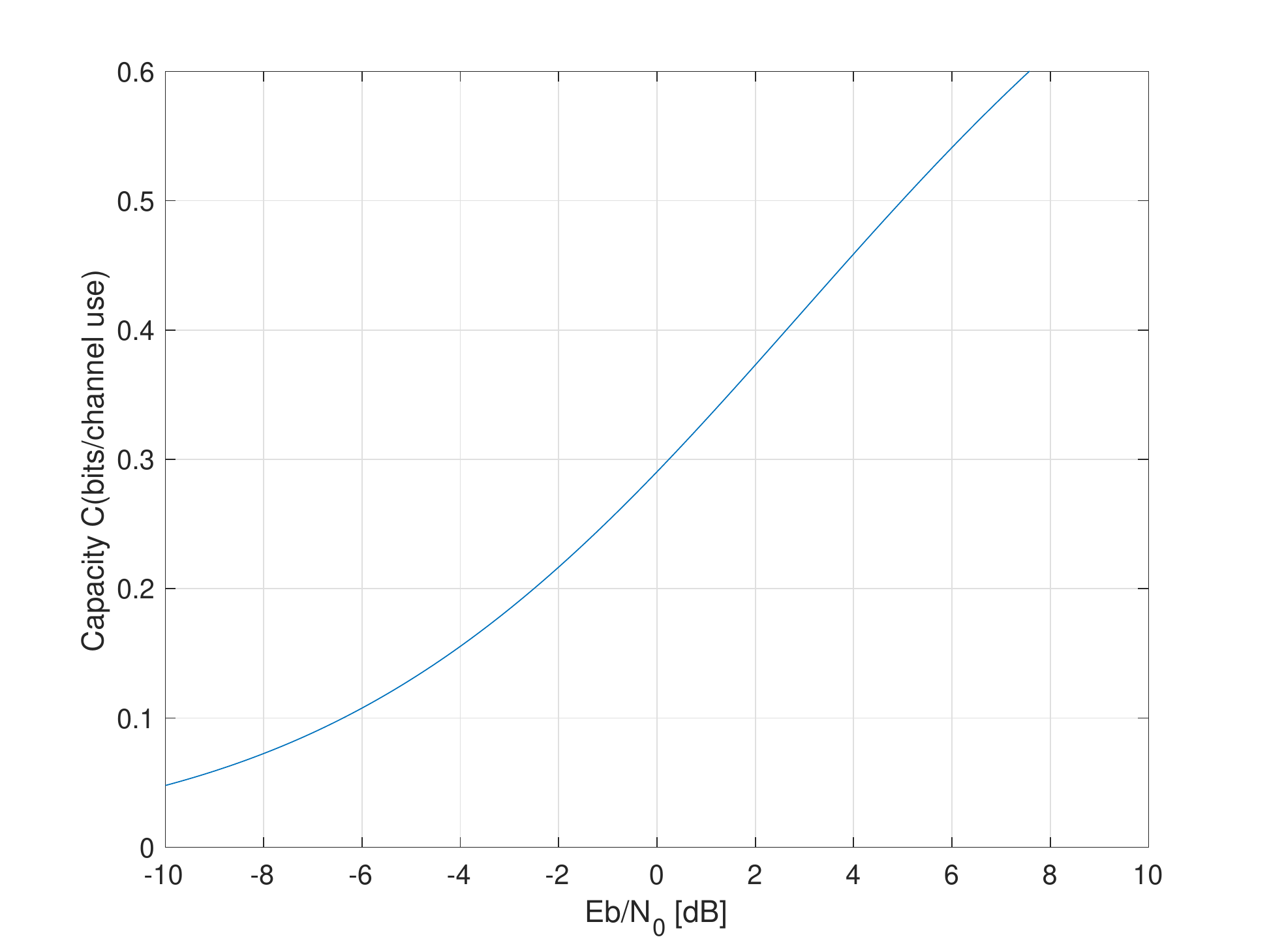}
\caption{The capacity of the one-bit quantized AWGN channel autoencoders in terms of ${\rm E_b/N_0}$. Here we observe for $C = \frac{1}{3}$ bits/use, $\gamma_{\rm min} = 1.051$ dB and for $C = \frac{1}{2}$ bits/use, $\gamma_{\rm min} = 4.983$ dB.}\label{fig:capacity}
\end{figure}

\subsection{Existence of the Global Minimum}
To achieve zero classification error above the minimum required SNR, the parameters of the encoder and decoder are trained such that the loss function converges to the global minimum. Next, we prove that there exists a global minima and at least one set of encoder-decoder parameters converges to this global minima.

\begin{theorem}\label{Theorem2}
For channel autoencoders, there is a global minima and at least one set of encoder-decoder pair parameters converges to this global minimum above the minimum required SNR. 
\end{theorem}
\begin{proof}
The depth and width of the neural layers in an autoencoder are determined beforehand, and these do not change dynamically. This means that $n$ and $M$ -- and hence the code rate -- is fixed. With sufficient SNR, one can ensure that this code rate is below the capacity, in which Shannon's coding theorem guarantees reliable (almost zero error) communication. To satisfy this for the autoencoder implementation of communication systems, the necessary and sufficient conditions in the proof of Shannon's channel coding theorem must be hold, which are (i) random code selection; (ii) jointly typical decoding; (iii) no constraint for unboundedly increasing the block length. It is straightforward to see that (i) is satisfied, because the encoder parameters are randomly initialized. Hence, the output of the encoder gives a random codeword. For (ii), Theorem \ref{Theorem1} shows that the aforementioned autoencoder results in maximum likelihood detection. Since maximum likelihood detection is a stronger condition than jointly typical decoding to make optimum detection, it covers the condition of jointly typical decoding and so (ii) is satisfied as well. For the last step, there is not any constraint to limit the width of the encoder layer. This means that (iii) is trivially met. Since channel autoencoders satisfy the Shannon's coding theorem, which states there is at least one good channel code to yield zero error communication, there exists a global minima that corresponds to the zero error communication, which can be achieved with at least one set of encoder-decoder parameters.
\end{proof}

It is not easy to converge to encoder-decoder parameters that result in global minimum due to the difficulties in training deep networks as mentioned previously. Additionally, the required one-hot coding in the architecture exponentially increases the input dimension, which renders it infeasible for practical communication systems, especially for high-dimensional communication signals. Thus, more practical autoencoder architectures are needed to design channel codes for one-bit quantization without sacrificing the performance.

\section{Practical Code Design for One-Bit Quantization} \label{Code Design}
To design a coding scheme under the constraint of one-bit ADCs for AWGN channels,  our approach -- motivated by Theorem \ref{Theorem1} -- is to make use of an autoencoder framework. Hence, we transform the code design problem for one-bit quantized AWGN channel to the problem of learning encoder-decoder pair for a special regularized autoencoder, in which the regularization comes from the one-bit ADCs and Gaussian noise. However, the one-hot encoding required by Theorem \ref{Theorem1} is not an appropriate method for high-dimensional communication signals, because this exponentially increases the input dimension while training neural networks. Another challenge is that one-bit quantization stymies gradient based learning for the layers before quantization, since it makes the derivative $0$ everywhere except at point $0$, which is not even differentiable. To handle all these challenges, we propose to train a practical but suboptimum autoencoder architecture and stack it with a state-of-the-art channel code that is designed for AWGN channels, but not for one-bit ADCs. The details of this design are elaborated next. In what follows, we justify the novelty of the proposed model in terms of machine learning principles.

\subsection{Autoencoder-Based Code Design}
To design a practical coding scheme for one-bit quantized communication, we need a practical (suboptimum) one-bit quantized AWGN channel autoencoder architecture. For this purpose, the one-bit quantized OFDM architecture proposed in \cite{BalAnd19} is modified for AWGN channels and implemented with time domain oversampling considering the pulse shape. This architecture is depicted in Fig. \ref{fig:one-bit-ae}, where the encoder includes the precoder, channel and equalizer. Note that there is a noise between the $l_0$ and $l_1$ layers that represents the noisy output of the equalizer. The equalized signal is further one-bit quantized, which corresponds to hard decision decoding, i.e., the decoder processes the signals composed of $\pm 1$. This facilitates training, which will be explained.
\begin{figure}[!t]
\centering
\subfigure[]{
\label{fig:one-bit-ae-a}
\includegraphics[width=3.75in]{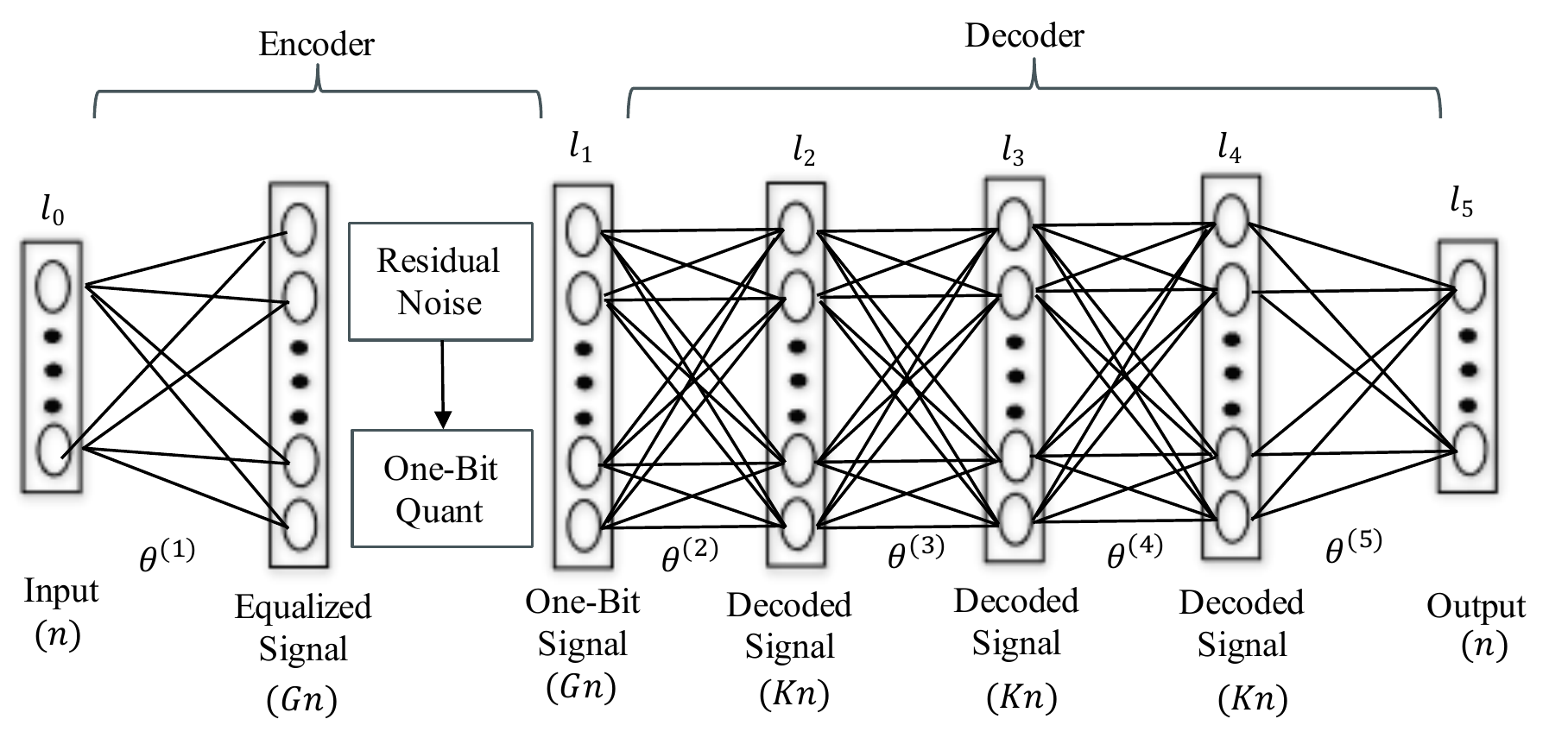}}
\qquad
\subfigure[]{
\label{fig:one-bit-ae-b}
\includegraphics[width=3.5in]{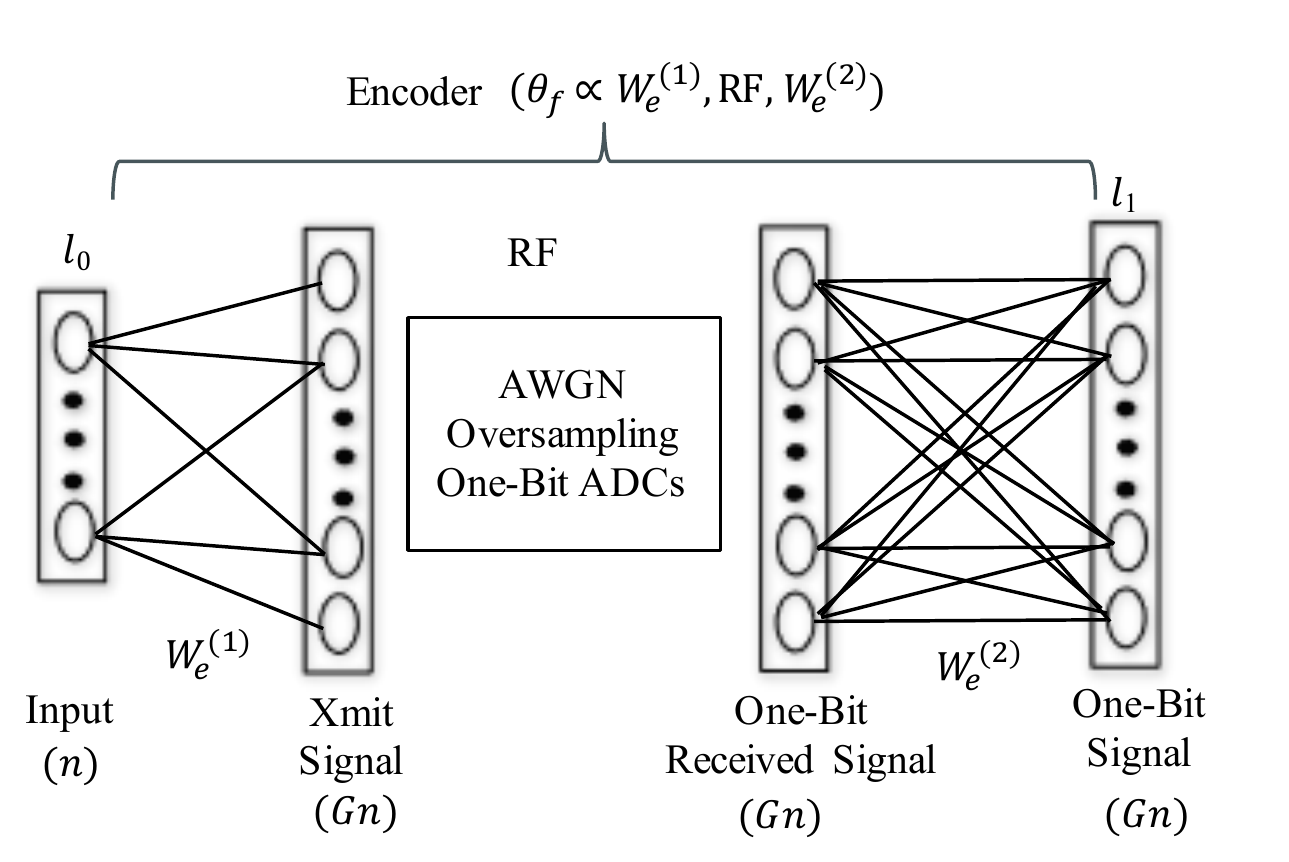}}
\caption{The one-bit quantized AWGN channel autoencoder that is trained in two-steps: (a) In the first step, only the decoder parameters $\theta_g = \{\theta^{(2)}, \cdots, \theta^{(5)}\}$ are trained (b) In the second step, the encoder parameters $\theta_f = \theta^{(1)} $ are trained.}
\label{fig:one-bit-ae}
\end{figure}

In this model, the binary valued input vectors are directly fed into the encoder without doing one-hot coding. This means that the input dimension is $n$ for $n$ bits. The key aspect of this architecture is to increase the input dimension by $G$ before quantization. This dimension is further increased by $K/G$, where $K>G$ while decoding the signal. Although it might seem that there is only one layer for the encoder in Fig. \ref{fig:one-bit-ae-a}, this in fact corresponds to the two neural layers and the RF part as detailed in Fig. \ref{fig:one-bit-ae-b}. The encoded signal is normalized to satisfy the transmission power constraint. There are $3$ layers in the decoder with the same dimension, in which the ReLU is used for activation. On the other hand, a linear activation function is used at the output, and the parameters are trained so as to minimize the mean square error between the input and output layer. Additionally, batch normalization is utilized after each layer to avoid vanishing gradients \cite{IoffeSzegedy15}. 

The two-step training policy is used to train the aforementioned autoencoder as proposed in \cite{BalAnd19}. Accordingly, in the first step shown in Fig. \ref{fig:one-bit-ae-a}, the decoder parameters are trained, whereas the encoder parameters $\theta_f$ are only randomly initialized, i.e., they are not trained due to the one-bit quantization. In the second step given in Fig. \ref{fig:one-bit-ae-b}, the encoder parameters are trained according to the trained and frozen decoder parameters by using the stored values of $l_0$ and $l_1$ layers in the first step in a supervised learning setup. Here, the precoder in the transmitter is determined by the parameters $W_e^{(1)}$. Then, the coded bits are transmitted using a pulse shaping filter $p(t)$ over an AWGN channel. In particular, these are transmitted with period $T/G$. In the receiver, the signal is processed with a matched filter $p^*(-t)$, oversampled by $G$, and quantized. This RF part corresponds to faster-than-Nyquist transmission, whose main benefit is to exploit the available excess bandwidth in the communication system. Notice that this transmission method is not employed in conventional codes, because it creates inter-symbol interference and leads to non-orthogonal transmission that degrades the tractability of the channel codes. The quantized signal is further processed by a neural layer or $W_e^{(2)}$ followed by another one-bit quantization so as to obtain the same $l_1$ layer in which the decoder parameters are optimized. The aim of the second one-bit quantization is to obtain exactly the same layer that the decoder expects, which would be impossible if the $l_1$ layer became a continuous valued vector. Since the decoder part of the autoencoder processes $\pm 1$, the proposed model can be considered as having a hard decision decoder.

The one-bit quantized AWGN channel autoencoder architecture apparently violates Theorem \ref{Theorem1} that assures the optimum coding, because neither one-hot coding nor softmax activation function is used. Additionally, ideal training is  not possible due to one-bit quantization. Thus, it does not seem possible to achieve almost zero error probability in detection with this suboptimum architecture and suboptimum training  even if $\gamma>\gamma_{\rm min}$. To cope with this problem, we propose employing an implicit regularizer that can serve as a priori information. More specifically, turbo coding is combined with the proposed autoencoder without any loss of generality, i.e., other off-the-shelf coding methods can also be used.  

The proposed coding scheme for AWGN channels under the constraint of one-bit ADC is given in Fig. \ref{fig:hybrid_code}, where the outer code is the turbo code and the inner code is the one-bit quantized AWGN channel autoencoder. In this concatenated code, the outer code injects strong a priori information for the inner code. Specifically, the bits are first coded with a turbo encoder for a given coding rate and block length. Then, the turbo coded bits in one block are divided into smaller subblocks, each of which is sequentially processed (or coded) by the autoencoder. In this manner, the autoencoder behaves like a convolutional layer by multiplying the subblocks within the entire block with the same parameters. Additionally,  dividing the code block into subblocks ensures reasonable dimensions for the neural layers. It is important to emphasize that the autoencoder does not consume further bandwidth. Rather, it exploits the excess bandwidth of the pulse shaping and packs the signal more intelligently by exploiting the sparsity in the autoencoder due to using ReLU, which means that nearly half of the input symbols are set to $0$ assuming that input is either $+1$ or $-1$ with equal probability. The double-coded bits (due to turbo encoder and autoencoder) are first decoded by the autoencoder. Then, the output of the autoencoder for all subblocks are aggregated and given to the outer decoder. 
\begin{figure} [!t] 
\centering 
\includegraphics [width=4in]{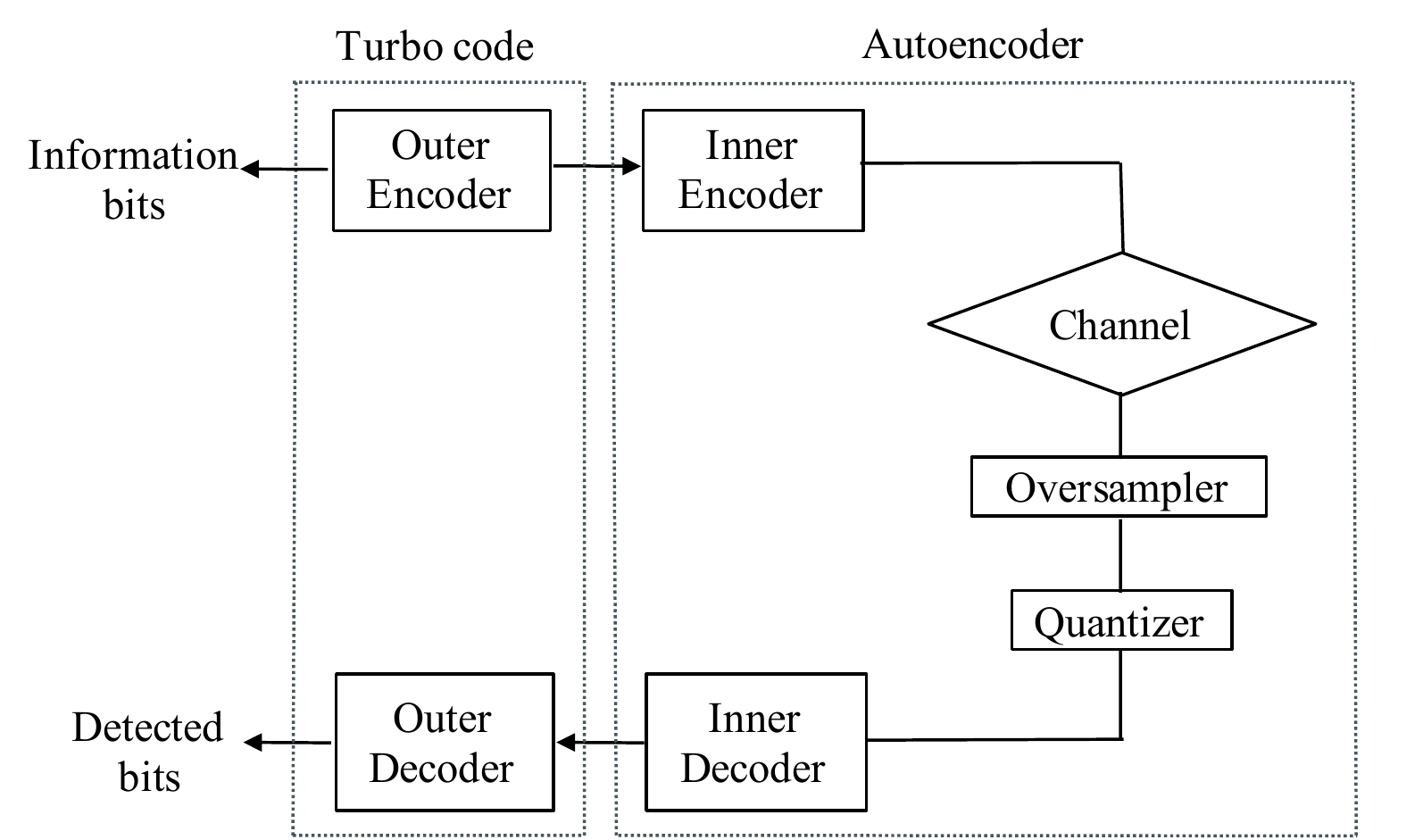}
\caption{The proposed concatenated code for the one-bit quantized AWGN channels, in which the outer code is the turbo code and the inner code is the autoencoder.}\label{fig:hybrid_code}
\end{figure} 

A concrete technical rationale for concatenating a turbo code and autoencoder is to provide Gaussian distributed data to the turbo decoder, which is optimized for AWGN and is known to perform very close to theoretical limits for Gaussian distributed data.  Below we formally prove that an autoencoder centered on the channel produces conditional Gaussian distributed data for the turbo decoder as in the case of AWGN channel even if there are some significant nonlinearities, such as one-bit quantization. 

\begin{theorem}\label{Theorem3}
The conditional probability distribution of the output of the autoencoder's decoder -- which is the input to the turbo decoder --conditioned on the output of the turbo encoder is a Gaussian process, despite the one-bit quantization at the front end of the receiver.
\end{theorem}
\begin{proof}
See Appendix \ref{appendix-C}.
\end{proof}

\begin{rem}
Theorem \ref{Theorem3} has important consequences, namely that even if there is a nonlinear operation in the channel or RF portion of the system, building an autoencoder around the channel provides a Gaussian distributed input to the decoder, and so standard AWGN decoders can be used without degradation. This brings robustness to the turbo codes against any nonlinearity in the channel: not just quantization but also phase noise, power amplifier nonlinearities, or nonlinear interference.
\end{rem}

\subsection{The Proposed Architecture as of relative to Deep Learning Principles}
Choosing some initial weights and moving through the parameter space in a succession of steps does not help to find the optimum solution in high-dimensional machine learning problems \cite{Bengio09}. Hence, it is very unlikely to achieve reliable communication by randomly initializing the encoder and decoder parameters and training these via gradient descent. This is particularly true if there is a non-differentiable layer in the middle of a deep neural network as in the case of one-bit quantization. Regularization is a remedy for such deep neural networks whose parameters cannot be initialized and trained properly. However, it is not clear what kind of regularizer should be utilized: it is problem-specific and there is not any universal regularizer. Furthermore, it is not easy to localize the impact of regularization from the optimization. To illustrate, in the seminal work of \cite{HintonTeh06} that successfully trains a deep network for the first time by pretraining all the layers and then stacking them together, it is not well understood whether the improvement is due to better optimization or better regularization \cite{BengioVincent13}.

Utilizing a novel implicit regularization inspired by coding theory has couple of benefits. First, it is applicable to many neural networks in communication theory: it is not problem-specific. Second, the handcrafted encoder can be treated as features extracted from another (virtual) deep neural network and combined with the target neural network. This means that a machine learning pipeline can be formed by stacking these two trained deep neural networks instead of stacking multiple layers. Although it is not known how to optimally combine the pretrained layers \cite{BengioVincent13}, it is much easier to combine two separate deep neural networks. Additionally, our model isolates the impact of optimization due to the one-bit quantization. This leads to a better understanding of the influence of regularization.

In deep neural networks, training the lower layers has the key role of determining the generalization capability \cite{Bengio09}. In our model, the lower layers can be seen as layers of a virtual deep neural network that can learn the state-of-the-art coding method. The middle layers are the encoder part of the autoencoder, which are the most problematic in terms of training ( due to one-bit quantization) and the higher layers are the decoder of the autoencoder. We find that even if the middle layers are suboptimally trained, the overall architecture performs well. That is, we claim that as long as the middle layers contribute to hierarchical learning, it is not important to optimally train their parameters. This brings significant complexity savings in training neural networks, but more work is needed to verify this claim more broadly.

\section{Numerical Results} \label{Numerical Results}
To determine the efficiency of the two-step training policy in the proposed autoencoder, how well the encoder parameters can be trained according to the decoder parameters is first shown. In what follows, the bit error rate (BER) of the proposed coding scheme is evaluated for QPSK and 16-QAM modulation under the constraint of one-bit quantization in the receivers for AWGN channels. In the simulations, a root raised cosine (RRC) filter with excess bandwidth $\alpha$ is considered as a pulse shape and $G$-fold time domain oversampling is utilized. Furthermore, $K$ is taken as $20$ without an extensive hyper-parameter search. We use the turbo code that is utilized in LTE \cite{LTERelease14}, which has a code rate of $\frac{1}{3}$ and a block length of $6144$. The codewords formed with this turbo code are processed with a subblock of length $64$ with the autoencoder. The proposed coding scheme is directly compared with this conventional turbo code in case of both unquantized (soft decision decoding) and one-bit quantized (hard decision decoding) samples. This is to explicitly show why an autoencoder is needed for one-bit quantization.

To observe how efficiently the encoder can be trained for the aforementioned training policy, its mean square error (MSE) loss function is plotted with respect to the number of epochs in Fig. \ref{fig:precoder_error}. As can be seen, the error goes to almost zero after a few hundred epochs. 
\begin{figure} [!t] 
\centering 
\includegraphics [width=5in]{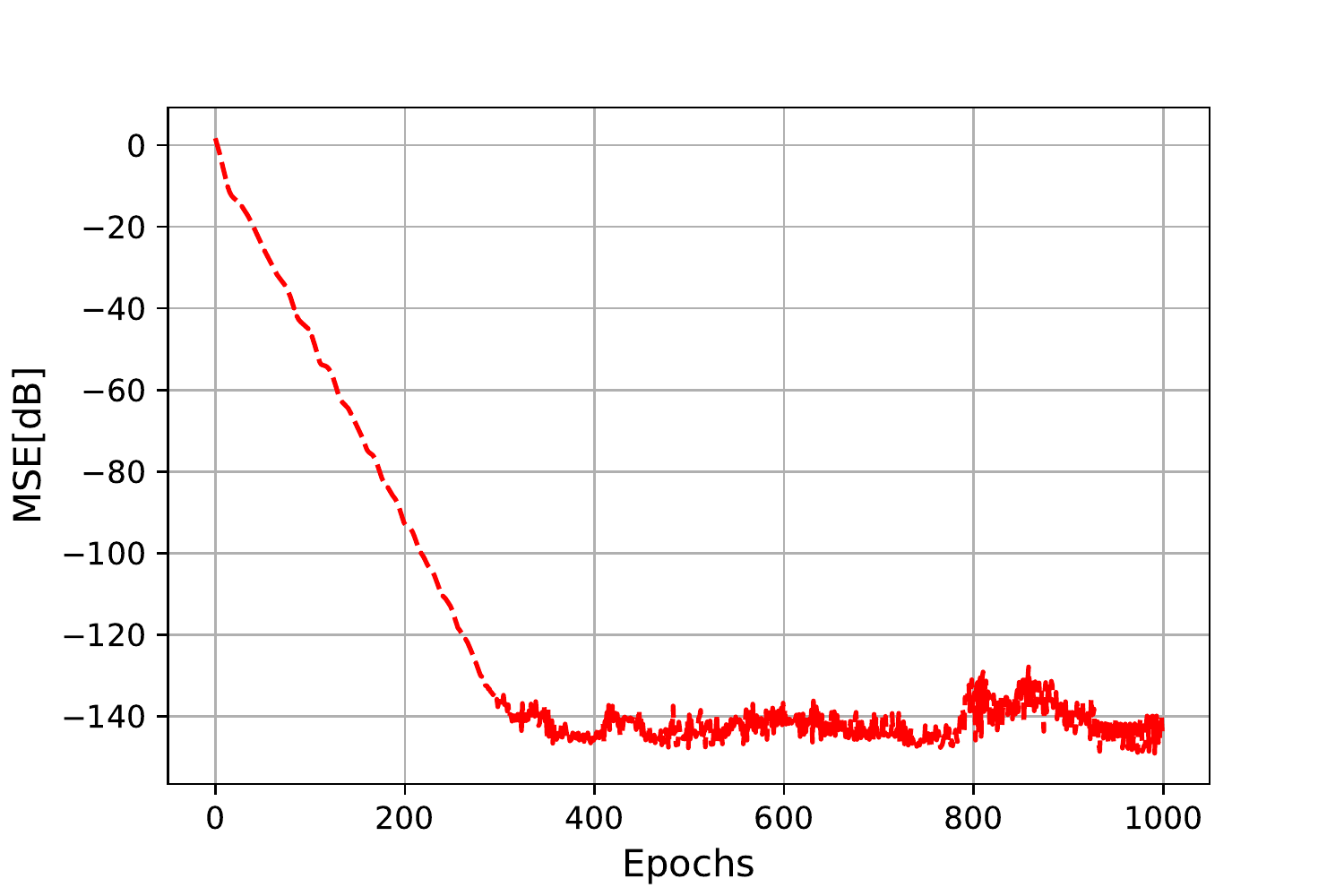}
\caption{The mean square error of the encoder parameters in the two-step training policy.}\label{fig:precoder_error}
\end{figure} 
One of the important observations in training the encoder is the behavior of the neural layer in the transmitter, which is the first layer in Fig. \ref{fig:one-bit-ae-b}. To be more precise, this layer demonstrates that nearly half of its hidden units (or neurons) become zero. This is due to the ReLU activation function and enables us to pack the symbols more intelligently. More precisely, the input of the autoencoder has $N$ units, and thus the dimension of the first hidden layer is $GN$, but only $GN/2$ of them have non-zero terms. Interestingly, the hidden units of this layer, which also correspond to the transmitted symbols, have quite different power levels from each other. To visualize this, when the all ones codeword with length $N = 64$ is given to the input of the autoencoder, the output of the first hidden layer for $G=2$ becomes as in Fig. \ref{fig:precoder_constellation}. According to that, $64$ neurons out of $128$ neurons become zero. Our empirical results also show that this is  independent of the value of $G$.
\begin{figure} [!t] 
\centering 
\includegraphics [width=5in]{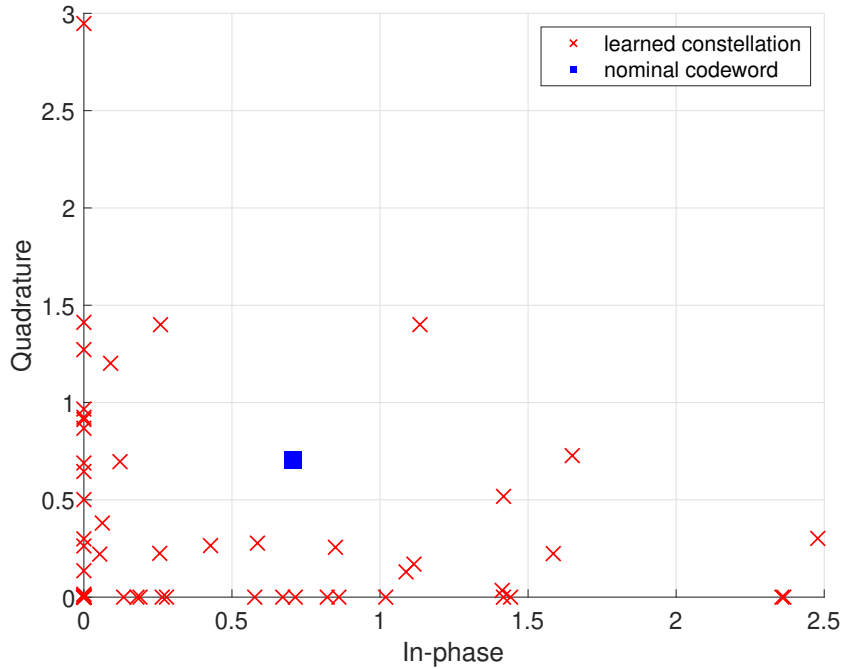}
\caption{The constellation of the transmitted signal learned from the all ones codeword. Note that all the points would be at $(0.707, 0.707)$ if there was not a precoder.}\label{fig:precoder_constellation}
\end{figure} 

In the proposed coding scheme, the symbols are transmitted faster with period $T/G$, however this does not affect the transmission bandwidth, i.e., the bandwidth remains the same \cite{LiverisCostas03}\footnote{Note that the complexity increase in the receiver due to faster-than-Nyquist transmission is not an issue for autoencoders, in which the equalizer is implemented as a neural network independent of transmission rate.}. 
Although the coding rate is $\frac{1}{G}$ in the proposed autoencoder, this does not mean that there is a trivial coding gain increase, because the bandwidth remains the same, and thus the minimum distance (or free distance) does not increase\footnote{In convolutional codes, the coding gain is smaller than or equal to $10\log_{10}({\rm coding\ rate} * {\rm minimum\ distance})$ \cite{Proakis}.}. The minimum distance can even decrease despite smaller coding rate, because dividing the same subspace into $G$ fold more partitions can decrease the distance between neighboring partitions.

To make a fair comparison between conventional turbo codes designed for orthogonal transmission with symbol period $T$ and our autoencoder-aided coding scheme, our methodology in the simulations is to consider the transmission rate increase as a bandwidth increase. In this manner, we first determine the maximum possible value of $G$ that corresponds to the available baseband bandwidth $(1+\alpha)/2T$. Then we add an SNR penalty by shifting the BER curve to the right, if the transmission rate increase exceeds the available bandwidth. To illustrate, if $\alpha$ is $1$, $G$ can be $4$ (because the non-orthogonal transmission symbol period becomes $T/2$ due to the precoder behavior) without needing to add an SNR penalty. However, if $G$ becomes $8$, the BER curve has to be shifted 3 dB to the right even if there is full excess bandwidth. This in fact explains what exploiting the excess bandwidth and packing the symbols more intelligently correspond to. 

The performance of the autoencoder-based concatenated code is given in Fig. \ref{fig:AEenpowered} for QPSK modulation. Specifically, when the turbo codes are decoded with $2$ iterations, the BER becomes as in Fig. \ref{fig:2iter}. Here, the proposed coding method can give very close performance to the turbo code that works with unquantized samples despite one-bit ADCs for $\alpha=1$. Although there is some performance loss for $\alpha=0.5$, our method can still outperform the turbo code that is optimized for one-bit samples. When $5$ iterations are employed for the turbo decoding, the gap due to the excess bandwidth increases a little as can be observed in Fig. \ref{fig:5iter}. Notice that our empirical results match with the derived expression in Lemma \ref{Lemma2}, which states the minimum required SNR for reliable communication. This also proves that the proposed suboptimum training policy can approach the performance of an autoencoder that is trained perfectly.
\begin{figure*}[!t]
\centering
\subfigure[2 iterations are used for turbo decoding.]{
\label{fig:2iter}
\includegraphics[width=5in]{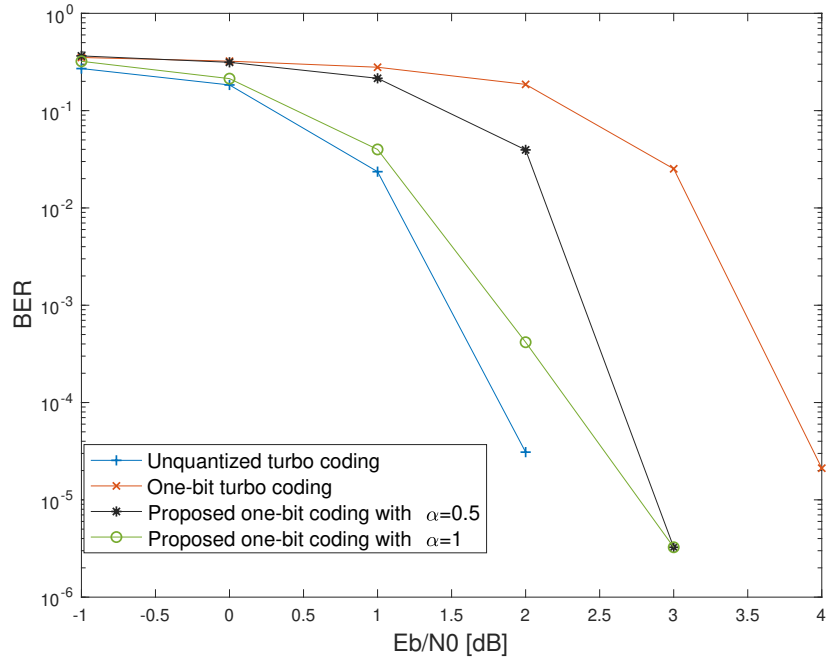}}
\qquad
\subfigure[5 iterations are used for turbo decoding.]{
\label{fig:5iter}
\includegraphics[width=5in]{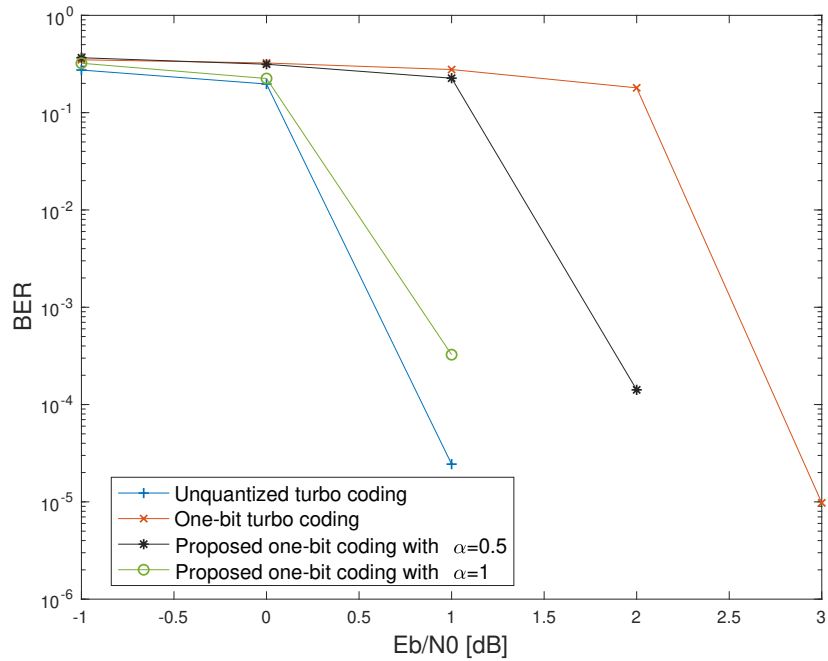}}
\caption{BER comparison of the proposed coding with the turbo coding that works with unquantized and one-bit quantized samples for QPSK modulation.}\label{fig:AEenpowered}
\end{figure*}

One-bit ADCs can work reasonably well in practice for QPSK modulation. However, this is not the case for higher order modulation, in which it is much more challenging to have a satisfactory performance with one-bit ADCs. To specify the benefit of the proposed coding scheme for higher order modulation, the simulation is repeated for $16$-QAM as depicted in Fig. \ref{fig:ber_vs_EbN0_5iter_16QAM}. As can be observed, the conventional turbo code is not sufficient for one-bit ADCs in case of $16$-QAM. On the other hand, the proposed coding method can give a similar waterfall slope with a nearly fixed SNR loss with respect to the turbo code that processes ideal unquantized samples. This result can be explained with Theorem \ref{Theorem3}. More precisely, in case of higher order modulations the nonlinearity stemming from one-bit ADCs considerably increases and deviates the AWGN channel to other non-Gaussian distributions. However, the inner code or the autoencoder produces a Gaussian process for the turbo decoder even if there is a high nonlinearity.
\begin{figure} [!h] 
\centering 
\includegraphics [width=4.75in]{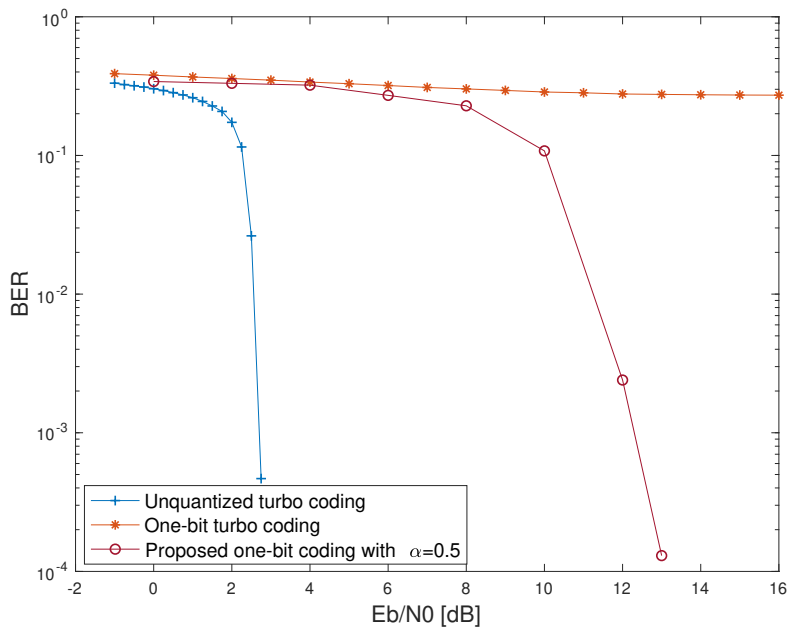}
\caption{BER performance for $16$-QAM.}\label{fig:ber_vs_EbN0_5iter_16QAM}
\end{figure} 

\section{Conclusions} \label{Conclusions}
In this paper, the development of hand-crafted channel codes for one-bit quantization is transformed into learning the parameters of a specially designed autoencoder. Despite its theoretical appeal, learning or training the parameters of an autoencoder is often very challenging. Hence, suboptimum training methods are needed that can lead to some performance loss. To compensate for this loss, we propose to use a state-of-the-art coding technique, which were developed according to the AWGN channel, as an implicit regularizer for autoencoders that are trained suboptimally. This idea is applied to design channel codes for AWGN channels under the constraint of one-bit quantization in receivers. Our results show that the proposed coding technique outperforms conventional turbo codes for one-bit quantization and can give performance close to unquantized turbo coding by packing the signal intelligently and exploiting the excess bandwidth. The superiority of the proposed coding scheme is more profound for higher order modulation in which one-bit ADCs are not previously viable even with powerful turbo codes. As future work, the idea of this hybrid code design can be extended to other challenging environments such as one-bit quantization for fading channels and high-dimensional MIMO channels. Additionally, it would be interesting to compensate for the performance loss observed in short block lengths for turbo, LDPC and polar codes with deep learning aided methods.

\appendices
\section{Proof of Theorem \ref{Theorem1}} \label{appendix-A}
In communication theory, solving (\ref{comm_theory}) for a given $n$, $M$ and SNR leads to the minimum probability of error, which can be achieved through maximum likelihood detection\footnote{We assume equal transmission probability of each message.}. Hence,
\begin{equation} \label{comm_enc_dec}
\epsilon_{ml}(n,M) =  \underset{f,g}{\mathrm{min\ }} \epsilon(n,M).
\end{equation}
It is straightforward to express  
\begin{equation}\label{comm_error}
\epsilon_{ml}(n,M)=\epsilon(n,M)
\end{equation} 
when $f=f^*$ and $g=g^*$. We need to prove that minimizing the loss function in (\ref{obj_fnc_ae}) while solving (\ref{deep_learning}) gives these same $f^*$ and $g^*$, i.e., $f^*=f_{AE}$ and $g^*=g_{AE}$.

Since the error probability is calculated message-wise instead of bit-wise in \eqref{error}, the $k$-dimensional binary valued input training vector $\textbf{s}$ is first encoded as a $2^k$-dimensional one-hot vector $\textbf{x}$ to form the messages, which is to say that $M=2^k$. Also, a softmax activation function is used to translate the entries of the output vector $\hat{\textbf{x}}$ into probabilities. With these definitions, the cross entropy function is employed to train the parameters\footnote{Here, we omit the subscript $c$ that represents the $c^{th}$ training sample for brevity.} 
\begin{equation}\label{cross_ent}
L(\textbf{x}, \hat{\textbf{x}}) = -\sum_{l=1}^{2^k}q[\hat{x}_l|\textbf{x}]\log(p[\hat{x}_l|\textbf{x}])
\end{equation}
where $q[\cdot | \cdot]$ is the empirical conditional probability distribution and $p[\cdot | \cdot]$ is the predicted conditional probability distribution (or the output of the neural network). 

Each output vector is assigned to only one of discrete $2^k$ classes, and hence the decision surfaces are $(2^k-1)$-dimensional hyperplanes for the $2^k$-dimensional input space. That is,
\begin{align} \label{hyperplanes}
    q[(\hat{x}_l)|\textbf{x}]= 
    \begin{cases}
       1 &  \text{$\textbf{x} \in \hat{x}_l$}\\
       0 & \text{o.w.}
    \end{cases} 
\end{align}
Substituting (\ref{hyperplanes}) in (\ref{cross_ent}) implies that
\begin{equation}\label{MAP_loss}
L(\textbf{x}, \hat{\textbf{x}}) = -\log(p[\hat{\textbf{x}}|\textbf{x}]). 
\end{equation}
It is straightforward to express that \eqref{MAP_loss} is minimized when $\mathcal{P}[\hat{\textbf{x}} = \textbf{x} |\textbf{x}]$ is maximized (or equivalently $\mathcal{P}[\hat{\textbf{x}} \neq \textbf{x} |\textbf{x}]$ is minimized). Since $\hat{\textbf{x}} = g(Y^n)$ and $\textbf{x}=i$,
\begin{equation}
\min \mathcal{P}[\hat{\textbf{x}}\neq \textbf{x} | \textbf{x}] = \epsilon_{ml}(n,M)
\end{equation}
due to \eqref{error} and \eqref{comm_enc_dec},  and is the case when $f=f^*$ and $g=g^*$ because of \eqref{comm_error}. This implies that 
\begin{equation}
(f^*,g^*) =   \underset{f,g}{\mathrm{arg\ min\ }} \mathcal{P}[\hat{\textbf{x}}\neq \textbf{x} | \textbf{x}] = \underset{f,g}{\mathrm{arg\ min\ }}L(\textbf{x}, \hat{\textbf{x}})
\end{equation}
By definition,
\begin{equation}\label{deep_learning_enc_dec}
(f_{AE},g_{AE}) = \underset{f,g}{\mathrm{arg\ min\ }} L(\textbf{x}, \hat{\textbf{x}}),
\end{equation}
and hence,
\begin{equation}\label{deep_learning_enc_dec}
(f^*,g^*) = (f_{AE},g_{AE}),
\end{equation}
which completes the proof due to the one-to-one mapping between $\textbf{s}$ and $\textbf{x}$, and $\hat{\textbf{s}}$ and $\hat{\textbf{x}}$.

\section{Proof of Theorem \ref{Theorem2}} \label{appendix-B}
The encoder parameters $\theta_f$ are initialized with zero-mean, unit variance Gaussian random variables in the one-bit quantized AWGN channel autoencoder. Hence, the mutual information is found over these random weights as
\begin{equation}
I(\textbf{s};\textbf{r}) = \mathbb{E}_{\theta_f}[I(\textbf{s};\textbf{r}|\theta_f)].
\end{equation}
By the definition of mutual information, 
\begin{equation}
\begin{split}
I(\textbf{s};\textbf{r}) &= \mathbb{E}_{\theta_f}[H(\textbf{r}|\theta_f)-H(\textbf{r}|\textbf{s},\theta_f)] \\
& =\mathbb{E}_{\theta_f}\left[\sum_{i=1}^nH(r_i|r_1, \cdots, r_{i-1},\theta_f)-H(r_i|r_1, \cdots, r_{i-1},\textbf{s},\theta_f)\right]. \\   
\end{split}
\end{equation}
The entries of the random matrix $\theta_f$ are i.i.d, and the noise samples are independent. This implies that the $r_i$'s are independent, i.e., 
\begin{equation}
I(\textbf{s};\textbf{r}) = n\mathbb{E}_{\theta_f}[H(r_i|\theta_f)-H(r_i|\textbf{s},\theta_f)]
\end{equation}
Since $r_i$ can be either $+1$ or $-1$ due to the one-bit quantization, $H(r_i)\leq 1$, which means 
\begin{equation}
\begin{split}
I(\textbf{s};\textbf{r}) &\leq n\mathbb{E}_{\theta_f}[1-H(r_i|\textbf{s},\theta_f)] \\
& \leq n\mathbb{E}_{\theta_f}\left[1+\sum_{\textbf{s}}\sum_{r_i}p[\textbf{s}, r_i|\theta_f]\text{log}(p[r_i|\textbf{s},\theta_f])\right] \\
& \leq n\mathbb{E}_{\theta_f}\left[1+\sum_{\textbf{s}}\sum_{r_i}p[r_i|\textbf{s},\theta_f]p[\textbf{s}]\text{log}(p[r_i|\textbf{s},\theta_f])\right]. \\
\end{split}
\end{equation}
Due to the one-to-one mapping between $\textbf{s}$ and $\textbf{x}$,
\begin{equation} \label{Bayes}
I(\textbf{x};\textbf{r}) \leq n\mathbb{E}_{\theta_f}\left[1+\sum_{\textbf{x}}\sum_{r_i}p[r_i|\textbf{x},\theta_f]p[\textbf{x}]\text{log}(p[r_i|\textbf{x},\theta_f])\right].
\end{equation}
Notice that for all $\textbf{x}$, only one of its elements is $1$, the rest are $0$. This observation reduces (\ref{Bayes}) to
\begin{equation}
I(\textbf{x};\textbf{r})  \leq n\mathbb{E}_{\theta_f}\left[1+\sum_{r_i}p[r_i|\textbf{x}=x,\theta_f]\text{log}(p[r_i|\textbf{x}=x,\theta_f])\right]
\end{equation}
where $x$ is one realization of $\textbf{x}$. Then, the total probability law gives
\begin{eqnarray}
I(\textbf{x};\textbf{r}) & \leq & \mathbb{E}_{\theta_f}[1 +p[r_i=+1|\textbf{x}=x,\theta_f]\text{log}(p[r_i=+1|\textbf{x}=x,\theta_f])\\ \nonumber
&& + p[r_i=-1|\textbf{x}=x,\theta_f]\text{log}(p[r_i=-1|\textbf{x}=x,\theta_f])].
\end{eqnarray}
Since
\begin{equation}
\begin{split}
p[r_i=+1|\textbf{x} = x,\theta_f] &= p[y_i\geq0|\textbf{x}=x,\theta_f]=Q\left(\theta_f\sqrt{\gamma}\right)\\
& =p[y_i<0|\textbf{x}=x,\theta_f]=1-Q\left(\theta_f\sqrt{\gamma}\right), \\   
\end{split}
\end{equation}
this completes the proof.

\section{Proof of Theorem \ref{Theorem3}} \label{appendix-C}
The autoencoder architecture, which is composed of 6 layers as illustrated in Fig. \ref{fig:one-bit-ae-a}, can be expressed layer-by-layer as
\begin{equation}
\begin{split}
l_0: & z^{(0)} = s, x^{(1)}  = \phi_0(z^{(0)}) = s \\   
l_1: & z^{(1)} = \theta^{(1)}x^{(1)} + b^{(1)}, x^{(2)}  = \mathcal{Q}(\phi_1(z^{(1)}) + n^{(1)}) \\   
l_2: & z^{(2)} = \theta^{(2)}x^{(2)} + b^{(2)}, x^{(3)}  = \phi_2(z^{(2)}) \\  
l_3: & z^{(3)} = \theta^{(3)}x^{(3)} + b^{(3)}, x^{(4)}  = \phi_3(z^{(3)}) \\  
l_4: & z^{(4)} = \theta^{(4)}x^{(4)} + b^{(4)}, x^{(5)}  = \phi_4(z^{(4)}) \\  
l_5: & z^{(5)} = \theta^{(5)}x^{(5)} + b^{(5)} \\
\end{split}
\end{equation}
where $\theta^{(l)}$ are the weights and $b^{(l)}$ is the bias. All the weights and biases are initialized with Gaussian random variables with variances $\sigma_\theta^2$ and $\sigma_b^2$, respectively, as is standard practice \cite{GlorotBengio10}, \cite{HeSun15}. Thus, $z_i^{(l)}|x^{(l)}$ is an identical and independent Gaussian process for every $i$ (or unit) with zero mean and covariance 
\begin{equation} \label{cov_l}
K^{(l)}(z,\hat{z}) = \sigma_b^2 + \sigma_\theta^2 \mathbb{E}_{z_i^{(l-1)} \sim \mathcal{N}(0,\,K^{(l-1)}(z,\hat{z}))}[\sigma_{l-1}(\phi(z_i^{(l-1)}))\sigma_{l-1}(\phi(\hat{z}_i^{(l-1)}))]
\end{equation}
where $\sigma_{l-1}(\cdot)$ is an identity function except for $l=2$, in which $\sigma_{1}(\cdot) = \mathcal{Q}(\cdot)$. As the width goes to infinity, \eqref{cov_l} can be written in integral form as 
\begin{equation} \label{cov_int}
\begin{split}
\underset{ {n^{(l-1)}\rightarrow \infty}}{\text{lim}}K^{(l)}(z,\hat{z}) = &\int\int \sigma_{l-1}(\phi_{l-1}(z_i^{(l-1)}))\sigma_{l-1}(\phi_{l-1}(\hat{z}_i^{(l-1)}))\\
&\mathcal{N}\left(z,\hat{z};0,\alpha_\theta^2\left[
  \begin{array}{cc}
  K^{(l-1)}(z,z) & K^{(l-1)}(z,\hat{z})  \\
  K^{(l-1)}(\hat{z},z) & K^{(l-1)}(\hat{z},\hat{z})   
  \end{array}
\right] + \alpha_b^2 \right)dzd\hat{z}. \\
\end{split}
\end{equation}
To be more compact, the double integral in \eqref{cov_int} can be represented with a function such that 
\begin{equation}
\underset{ {n^{(l-1)}\rightarrow \infty}}{\text{lim}}K^{(l)}(z,\hat{z}) = F_{l-1}(K^{(l-1)}(z,\hat{z})).
\end{equation}
Hence, $z^{(5)}|s$ is a Gaussian process with zero mean and covariance 
\begin{equation}
K^{(5)}(z,\hat{z}) = F_4(\cdots (F_1(K^{(1)}(z,\hat{z}))))
\end{equation}
when $\min(n_1, \cdots, n_5) \rightarrow \infty$, i.e., the output of the autoencoder yields Gaussian distributed data in the initialization phase.

During training, the parameters are iteratively updated as
\begin{equation}
\Theta_n = \Theta_{n-1}-\eta \nabla_{\Theta_{n-1}} L(\Theta_{n-1})
\end{equation}
where $\Theta_n = \{\theta_n^{(1)}, \cdots, \theta_n^{(5)}, b^{(1)}, \cdots, b^{(5)} \}$, and $L(\cdot)$ is the loss function. In parallel, the output $z^{(5)}$ is updated as 
\begin{equation} \label{out_upd}
z^{(5)}_n = z^{(5)}_{n-1} + \nabla_{\Theta_{n-1}} (z^{(5)}_{n-1})(\Theta_n - \Theta_{n-1}).
\end{equation}
The gradient term in \eqref{out_upd} is a nonlinear function of the parameters. Nevertheless, it was recently proven in \cite{JaehoonLee19} that as the width goes to infinity, this nonlinear term can be linearized via a first-order Taylor expansion. More precisely,
\begin{equation} \label{out_ae}
z^{(5)}_n = z^{(5)}_{0} + \nabla_{\Theta_{0}}(z^{(5)}_{0})(\Theta_{n}-\Theta_0) + \mathcal{O}((\min(n_1, \cdots, n_5)^{-0.5} )
\end{equation}
where the output at the initialization or $z^{(5)}_{0}$ is Gaussian as discussed above. Since the gradient (and hence the Jacobian matrix) is a linear operator, and a linear operation on a Gaussian process results in a Gaussian process, the output of the autoencoder for a given input (or $z^{(5)}_n|s$) is a Gaussian process throughout training with gradient descent.

\end{document}